\documentclass[journal]{IEEEtran}[12pt]
\ifCLASSINFOpdf
  \usepackage[pdftex]{graphicx}
  \graphicspath{{../pdf/}{../jpeg/}}
  \DeclareGraphicsExtensions{.pdf,.jpeg,.png}
\else
  \usepackage[dvips]{graphicx}
  \graphicspath{{../eps/}}
  \DeclareGraphicsExtensions{.pdf}
\fi
\usepackage{epstopdf}
\usepackage{stfloats}
\usepackage[caption=false, font=footnotesize]{subfig}
\usepackage[cmex10]{amsmath}
\usepackage{array}
\usepackage{mdwmath}
\usepackage{mdwtab}
\usepackage{graphicx}
\usepackage{setspace}
\usepackage{amssymb,amsmath,amsthm,amsfonts}
\usepackage{color}
\usepackage{hyperref} 
\usepackage{threeparttable}
\usepackage{algpseudocode}
\usepackage{algorithmicx}
\usepackage{algorithm}  
\usepackage{cite}
\newtheorem{lemma}{Lemma}

\allowdisplaybreaks[4]
\begin{document}

\title{{Downlink Performance Analysis of Pinching Antenna Systems: WDMA or NOMA?} \\
}

\author{Han~Zhang,~\IEEEmembership{Student Member,~IEEE,}
	Bingxin~Zhang,~\IEEEmembership{Member,~IEEE,}
	Yizhe~Zhao,~\IEEEmembership{Member,~IEEE,}
	and Kun~Yang,~\IEEEmembership{Fellow,~IEEE}
	\thanks{Han Zhang, Bingxin Zhang and Kun Yang are with the State Key Laboratory of Novel Software Technology, Nanjing University, Nanjing, 210008, China, Institute of Intelligent Networks and Communications (NINE), Collaborative Innovation Center of Novel Software Technology and Industrialization, and School of Intelligent Software and Engineering, Nanjing University (Suzhou Campus), Suzhou, 215163, China (email: hanzhangl@smail.nju.edu.cn; bxzhang@nju.edu.cn; kunyang@nju.edu.cn).}
	\thanks{Yizhe Zhao is  with the School of Information and Communication Engineering, University of Electronic Science and Technology of China, Chengdu 611731, China (e-mail: yzzhao@uestc.edu.cn).}
}

\maketitle

\begin{abstract}
This paper presents an analytical framework for downlink pinching antenna systems (PASS) employing waveguide division multiple access (WDMA) and non-orthogonal multiple access (NOMA). A unified channel model is developed to capture antenna deployment, user spatial distribution, and path loss. Closed-form and single-integral expressions for the outage probability and average achievable rate are derived and validated via Monte Carlo simulations. The results show that NOMA achieves higher spectral efficiency at high transmit signal-to-noise ratio (SNR) due to successive interference cancellation (SIC), whereas WDMA offers more reliable performance at low to moderate SNR but suffers from an outage floor and rate saturation at high SNR. Moreover, WDMA performance is more sensitive to the user spatial distribution due to the spatially dependent inter-waveguide interference. These findings provide design insights for access-scheme selection and antenna placement in PASS.

\end{abstract}


\section{Introduction}
{\huge T}he continuous evolution of wireless communication systems demands flexible and high-performance technologies to support emerging applications such as extended reality, smart healthcare, autonomous systems, and the metaverse. As a result, flexible antenna systems that can reconfigure wireless propagation conditions in real time have emerged as a promising solution~\cite{yang2025flexible}. Representative examples include reconfigurable intelligent surfaces (RIS)~\cite{pan2021reconfigurable}, fluid antennas (FAs)~\cite{11297439}, and movable antennas (MAs)~\cite{zhu2023movable}, which physically reposition antenna elements or reconfigure reflection coefficients to achieve favorable channel conditions.

However, these paradigms still face inherent limitations in mitigating large-scale propagation loss. RIS-assisted links suffer from the double large-scale fading effect~\cite{pan2021reconfigurable}, while fluid antenna systems and movable antenna systems mainly combat small-scale fading but remain constrained by the dominant free-space path loss~\cite{11297439,zhu2023movable}. To directly control large-scale propagation loss, pinching antenna systems (PASS) have been recently proposed as a waveguide-assisted flexible-antenna architecture~\cite{ding2025flexible}. By selectively activating pinching antennas (PAs) along dielectric waveguides, PASS can reduce propagation loss and adapt link geometry on demand. Building upon this architecture, existing studies have investigated PA position optimization and joint beamforming designs to improve coverage, system throughput, and multiuser performance~\cite{wang2025antenna,chen2025dynamic}. To further enhance multiuser scalability, waveguide division multiple access (WDMA)~\cite{zhao2025waveguide} and non-orthogonal multiple access (NOMA)~\cite{ren2025pinching} have been separately investigated in PASS.

Despite these advances, a unified analytical framework for comparing WDMA and NOMA in PASS under random user locations and spatial interference coupling is still lacking. Existing WDMA- and NOMA-related PASS studies primarily focus on optimization-oriented designs, leaving the analytical characterization of these schemes largely open. In particular, analytically tractable performance characterizations that jointly capture the effects of antenna deployment, user spatial distributions, and access-scheme-dependent interference have not yet been established.

To fill this gap, this paper develops an analytical framework to characterize the downlink performance of PASS with WDMA and NOMA. A unified channel model is established to capture antenna deployment and user spatial distribution, based on which closed-form and single-integral expressions for the outage probability and average achievable rate are derived and validated by Monte Carlo simulations. The results demonstrate that WDMA provides more reliable performance at low-to-moderate transmit SNR and with larger user separation, while NOMA achieves higher spectral efficiency at high transmit SNR owing to successive interference cancellation. Furthermore, by explicitly quantifying the impact of PA deployment height on large-scale fading and interference, the proposed analysis offers practical guidelines for selecting suitable multiple-access schemes under different transmit power levels and antenna deployment configurations.

\section{System Model}
\begin{figure}[!t]
  \centering
  \includegraphics[width=0.9\linewidth]{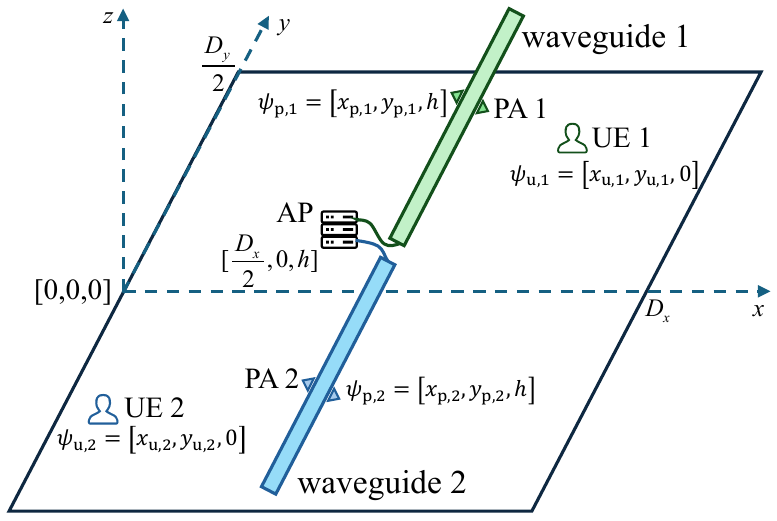}
  \caption{System model of PASS.}
  \label{fig:sys_models}
\end{figure}

This section presents the downlink system models of PASS employing WDMA and NOMA, as illustrated in Fig.~\ref{fig:sys_models}. 
WDMA and NOMA are studied under the same PASS geometry, service region, and LoS path-loss-based channel model, but with access-scheme-dependent UE indexing. Specifically, WDMA indexes UEs by their associated sub-regions and waveguides, whereas NOMA indexes UEs according to the near--far channel ordering required for SIC.

\subsection{WDMA}
We consider a downlink PASS with WDMA, as shown in Fig.~\ref{fig:sys_models}. 
The access point (AP) is connected to two distinct dielectric waveguides through independent feed links, where waveguide~$i$ is dedicated to UE~$i$, $i\in\{1,2\}$. 
On each waveguide, a single pinching antenna, denoted by PA~$i$, is deployed at height $h$ to serve its associated UE. 
The coordinates of UE~$i$ and PA~$i$ are denoted by $\boldsymbol{\psi}_{\mathrm{u},i}=[x_{u,i},y_{u,i},0]$ and $\boldsymbol{\psi}_{\mathrm{p},i}=[D_x/2,y_{u,i},h]$, respectively, such that PA~$i$ is vertically aligned with UE~$i$ to minimize the propagation distance.

The UEs are uniformly distributed within two non-overlapping sub-regions along the $y$-axis, i.e.,
\begin{align}
x_{u,i} \sim \mathcal{U}(0,D_x), \quad y_{u,1} \sim \mathcal{U}(0,D_y), \quad y_{u,2} \sim \mathcal{U}(-D_y,0).
\end{align}

The LoS channel from PA $i$ to UE $i$ follows the free-space path-loss model \cite{zhao2025waveguide}
\begin{align}
g_i
= \frac{\sqrt{\eta}\,
e^{-j\frac{2\pi}{\lambda}\sqrt{(x_{u,i}-D_x/2)^2+h^2}}}
{\sqrt{(x_{u,i}-D_x/2)^2+h^2}},
\end{align}
where $\eta=c_l^2/(16\pi^2 f_c^2)$ denotes the free-space path-loss factor, 
and $c_l$ is the speed of light. Similarly, the interfering channel from PA $i'$ to UE $i$ $(i'\neq i)$ is
\begin{align}
g_{i,i'}
= \frac{\sqrt{\eta}\,
e^{-j\frac{2\pi}{\lambda}\sqrt{(x_{u,i}-D_x/2)^2+h^2+(y_{u,i}-y_{u,i'})^2}}}
{\sqrt{(x_{u,i}-D_x/2)^2+h^2+(y_{u,i}-y_{u,i'})^2}}.
\end{align}

Assuming equal power allocation, the received signal at UE~$i$ is
\begin{align}
y_i 
&= e^{-j\frac{2\pi}{\lambda_{\rm g}}|y_{u,i}|}
g_i\sqrt{\tfrac{P}{2}}\,s_i
+ e^{-j\frac{2\pi}{\lambda_{\rm g}}|y_{u,i'}|}
g_{i,i'}\sqrt{\tfrac{P}{2}}\,s_{i'} 
+ n_i,
\end{align}
where $\lambda_{\rm g}$ denotes the wavelength in the dielectric waveguide and $n_i\sim\mathcal{CN}(0,\sigma_i^2)$. 
Here, $|y_{u,i}|$ and $|y_{u,i'}|$ represent the propagation distances from the AP feed point to the corresponding PAs along the waveguides. Since the waveguide phase factors have unit magnitude, the corresponding SINR is
\begin{align}
\gamma_i
= 
\frac{
\left|
e^{-j\frac{2\pi}{\lambda_{\rm g}}|y_{u,i}|} g_i
\right|^2
}{
\left|
e^{-j\frac{2\pi}{\lambda_{\rm g}}|y_{u,i'}|} g_{i,i'}
\right|^2
+
\tfrac{2\sigma_i^2}{P}
}
=
\frac{|g_i|^2}{|g_{i,i'}|^2 + \tfrac{2\sigma_i^2}{P}},
\label{eq:WDMA_SINR}
\end{align}

\subsection{NOMA}
We next consider a downlink PASS employing NOMA, as illustrated in Fig.~\ref{fig:sys_models}. 
Unlike WDMA, a single PA simultaneously serves two UEs via power-domain superposition coding. 
The PA is vertically aligned with the UE closest to the waveguide, denoted as UE~1, while the farther UE is denoted as UE~2. This placement reduces the large-scale path loss of the aligned UE and provides a clear decoding order for SIC. It also keeps the PA placement principle consistent with WDMA, so that NOMA and WDMA can be compared under the same PASS geometry by changing the activated PAs or waveguides and the transmitted signal structure.
Specifically, let $x_1,x_2\sim\mathcal{U}(0,D_x)$; then
\begin{align}
x_{u,1}&=\arg\min_{x\in\{x_1,x_2\}} |x-D_x/2|,\notag\\
x_{u,2}&=\arg\max_{x\in\{x_1,x_2\}} |x-D_x/2|.
\end{align}
The AP transmits $s=\sqrt{\alpha_1P}\,s_1+\sqrt{\alpha_2P}\,s_2$, where 
$\alpha_1+\alpha_2=1$ and $\alpha_1<\alpha_2$.

Following the LoS free-space path-loss model 
\cite{zhao2025waveguide,ren2025pinching}, the PA--UE channel is
\begin{align}
g_i=\frac{\sqrt{\eta}\,e^{-j\frac{2\pi}{\lambda}d_i}}{d_i},
\end{align}
where $d_1=\sqrt{(x_{u,1}-D_x/2)^2+h^2}$ and 
$d_2=\sqrt{G_2+A^2+h^2}$, with 
$G_i\triangleq(x_{u,i}-D_x/2)^2$, $A\triangleq|y_{u,2}-y_p|$, and $y_p=y_{u,1}$. 
The received signal at UE~$i$ is
\begin{align}
y_i=e^{-j\frac{2\pi}{\lambda_{\rm g}}|y_p|}g_i s+n_i,
\end{align}
where $n_i\sim\mathcal{CN}(0,\sigma_i^2)$. 
Since the waveguide phase factor has unit magnitude and is common to both superimposed components, it does not affect the SIC or SINR analysis. Therefore,

\begin{align}
\gamma_1=\frac{|e^{-j\frac{2\pi}{\lambda_{\rm g}}|y_p|} g_1|^2\alpha_1 P}{\sigma_1^2}
=\frac{\eta\alpha_1 P}{\sigma_1^2\!\left[(x_{u,1}-D_x/2)^2+h^2\right]},
\end{align}
and

\begin{align}
\gamma_2
=\frac{|e^{-j\frac{2\pi}{\lambda_{\rm g}}|y_p|} g_2|^2\alpha_2 P}{|e^{-j\frac{2\pi}{\lambda_{\rm g}}|y_p|} g_2|^2\alpha_1 P+\sigma_2^2}
=\frac{\eta\alpha_2 P}{\eta\alpha_1 P+\sigma_2^2(G_2+A^2+h^2)}.
\end{align}
This model highlights that WDMA exploits spatial separation across multiple waveguides, whereas NOMA relies on power-domain multiplexing and SIC at a single PA.


\section{Performance analysis of downlink transmission}
\subsection{WDMA}
In this subsection, we analyze the downlink performance of the proposed WDMA-enabled PASS. Our focus is placed on two key performance metrics, namely the outage probability and the average achievable rate of UE 1. The results for UE 2 follow by symmetry. 
\subsubsection{Outage Probability}
By defining $G(x)\triangleq (x-D_x/2)^2+h^2$ and 
$Y\triangleq(y_{u,1}-y_{p,2})^2$, substituting~\eqref{eq:WDMA_SINR} 
into the outage event $\gamma_1\le\gamma_{\rm th}$ gives
$
\frac{1}{G(x_{u,1})}
\le
\gamma_{\rm th}
\left[
\frac{1}{G(x_{u,1})+Y}
+
\frac{2\sigma_1^2}{\eta P}
\right].$
Solving this inequality with respect to $Y$ yields $Y\le T(x_{u,1})$, where
$
T(x)=
\frac{G(x)}
{\tfrac{1}{\gamma_{\mathrm{th}}}
-\tfrac{2\sigma_1^2}{\eta P}G(x)}
-G(x).$

\begin{lemma}
The outage probability of UE~1 in WDMA can be approximated by Gaussian--Chebyshev quadrature \cite{kumar2022performance} as
\begin{align}
P_{\mathrm{out}}
\approx
\frac{\pi}{2N}
\sum_{k=1}^{N}
\sqrt{1-t_k^2}\,
p_{\mathrm{out}}\!\left(\frac{D_x}{2}(t_k+1)\right),
\label{eq:WDMA_outage_GC}
\end{align}
where $t_k=\cos\frac{(2k-1)\pi}{2N}$, $N$ is the number of quadrature nodes, and
\begin{align}
p_{\mathrm{out}}(x)=
\begin{cases}
1, &
\tfrac{1}{\gamma_{\mathrm{th}}}
-\tfrac{2\sigma_1^2}{\eta P}G(x)\le 0,\\[5pt]
F_Y\!\big(T(x)\big), & \text{otherwise},
\end{cases}
\end{align}
with
\begin{align}
F_Y(y)=
\begin{cases}
0, & y\le 0,\\
\dfrac{y}{2D_y^2}, & 0<y\le D_y^2,\\[4pt]
\dfrac{4D_y\sqrt{y}-y}{2D_y^2}-1, & D_y^2<y\le 4D_y^2,\\[4pt]
1, & y>4D_y^2.
\end{cases}
\end{align}
\end{lemma}

\begin{proof}
Substituting \eqref{eq:WDMA_SINR} into the outage event 
$\gamma_1\le\gamma_{\rm th}$ gives
\begin{align}
\frac{1}{G(x_{u,1})}
\le
\gamma_{\rm th}
\left[
\frac{1}{G(x_{u,1})+Y}
+
\frac{2\sigma_1^2}{\eta P}
\right].
\end{align}
After rearranging, we have
\begin{align}
\frac{1}{\gamma_{\rm th}}
-
\frac{2\sigma_1^2}{\eta P}G(x_{u,1})
\le
\frac{G(x_{u,1})}{G(x_{u,1})+Y}.
\end{align}
If the left-hand side is non-positive, the above inequality always holds, and thus the conditional outage probability equals one. Otherwise, solving the inequality with respect to $Y$ yields
\begin{align}
Y
\le
\frac{G(x_{u,1})}
{\tfrac{1}{\gamma_{\rm th}}
-\tfrac{2\sigma_1^2}{\eta P}G(x_{u,1})}
-
G(x_{u,1})
=
T(x_{u,1}).
\end{align}

Since $y_{u,1}$ and $y_{p,2}$ are uniformly distributed over two adjacent intervals of length $D_y$, their difference follows a triangular distribution over $(0,2D_y]$. Hence, for $y\ge0$,
\begin{align}
F_Y(y)
=
\Pr\{Y\le y\}
=
\int_0^{\sqrt{y}} f_U(u)\,du,
\end{align}
which gives the piecewise CDF of $Y$ in the lemma.

Since $x_{u,1}$ is independent of $Y$, applying the law of total expectation gives
\begin{align}
P_{\mathrm{out}}
&=
\mathbb{E}_{x_{u,1}}
\!\left[
\Pr\{Y\le T(x_{u,1})\mid x_{u,1}\}
\right] \notag\\
&=
\frac{1}{D_x}
\int_0^{D_x}
\Pr\{Y\le T(x)\}\,dx.
\end{align}
When the denominator of $T(x)$ is non-positive, the conditional outage probability is one. Otherwise, the inner probability is obtained by substituting $T(x)$ into the CDF of $Y$, i.e.,
\begin{align}
\Pr\{Y\le T(x)\}=F_Y\!\big(T(x)\big).
\end{align}
Hence,
\begin{align}
P_{\mathrm{out}}
=
\frac{1}{D_x}
\int_0^{D_x}
p_{\mathrm{out}}(x)\,dx.
\end{align}

Finally, applying the transformation
\begin{align}
x=\frac{D_x}{2}(t+1),\qquad dx=\frac{D_x}{2}dt,
\end{align}
we obtain
\begin{align}
P_{\mathrm{out}}
=
\frac{1}{2}
\int_{-1}^{1}
p_{\mathrm{out}}\!\left(\frac{D_x}{2}(t+1)\right)dt.
\end{align}
Using the Gaussian--Chebyshev quadrature rule
\begin{align}
\int_{-1}^{1} f(t)dt
\approx
\frac{\pi}{N}
\sum_{k=1}^{N}
\sqrt{1-t_k^2}\,f(t_k),
\end{align}
with $t_k=\cos\frac{(2k-1)\pi}{2N}$, we obtain \eqref{eq:WDMA_outage_GC}.
\end{proof}

\subsubsection{Average Achievable Rate}
\begin{lemma}
Let $B\triangleq 2\sigma_1^2/(\eta P)$. The average achievable rate of UE~1 in WDMA can be efficiently approximated by Gaussian--Chebyshev quadrature \cite{kumar2022performance} as
\begin{align}
R_{\mathrm{avg}}
\approx
\frac{\pi}{2ND_y^2\ln2}
\sum_{k=1}^{N}\sqrt{1-t_k^2}\,
\Phi\!\left(\frac{D_x}{2}(t_k+1)\right),
\label{eq:WDMA_rate_GC}
\end{align}
where $t_k=\cos\frac{(2k-1)\pi}{2N}$, $N$ is the number of quadrature nodes, and
\begin{align}
\Phi(x)
&\triangleq
2P_1(D_y;x)-P_1(2D_y;x) \notag\\
&\quad
+2D_y\big[P_0(2D_y;x)-P_0(D_y;x)\big].
\end{align}
Here,
\begin{align}
P_0(y;x)
&=
\big[
\mathcal{J}_0(u;a(x),b(x))
-\mathcal{J}_0(u;c(x),d(x))
\big]_{0}^{y},\\
P_1(y;x)
&=
\big[
\mathcal{J}_1(u;a(x),b(x))
-\mathcal{J}_1(u;c(x),d(x))
\big]_{0}^{y},
\end{align}
with
\begin{align}
a(x)&=2G(x)+BG^2(x),\quad b(x)=1+BG(x),\notag\\
c(x)&=G(x)+BG^2(x),\quad d(x)=BG(x),
\end{align}
where $\mathcal{J}_0(\cdot)$ and $\mathcal{J}_1(\cdot)$ are obtained from the standard logarithmic integral identities in \cite[(2.733.1), (2.733.2)]{edition2007table}, and are given by
\begin{align}
\mathcal{J}_0(u;a,b)
&=u\ln(a+bu^2)-2u
+2\sqrt{\frac{a}{b}}\arctan\!\left(u\sqrt{\frac{b}{a}}\right),\\
\mathcal{J}_1(u;a,b)
&=\frac{(a+bu^2)\ln(a+bu^2)-(a+bu^2)}{2b}.
\end{align}
\end{lemma}

\begin{proof}
From \eqref{eq:WDMA_SINR}, the instantaneous achievable rate of UE~1 can be written as
\begin{align}
R_1
&=\log_2(1+\gamma_1) \notag\\
&=
\log_2\!\left(
1+\frac{U^2+G(x)}
{G(x)\big(1+B(U^2+G(x))\big)}
\right).
\end{align}
After rearranging the logarithmic term, we have
\begin{align}
R_1
=
\frac{1}{\ln2}
\ln
\frac{a(x)+b(x)u^2}
{c(x)+d(x)u^2}.
\end{align}

Since $y_{u,1}$ and $y_{p,2}$ are uniformly distributed over two adjacent intervals of length $D_y$, the difference $U=y_{u,1}-y_{p,2}$ follows the triangular PDF
\begin{align}
f_U(u)=
\begin{cases}
\dfrac{u}{D_y^2}, & 0<u\le D_y,\\[4pt]
\dfrac{2D_y-u}{D_y^2}, & D_y<u\le 2D_y,\\[4pt]
0, & \text{otherwise}.
\end{cases}
\end{align}
Thus, averaging $R_1$ over $x_{u,1}$ and $U$ gives
\begin{align}
R_{\mathrm{avg}}
&=
\frac{1}{D_xD_y^2\ln2}
\int_0^{D_x}
\Bigg[
\int_0^{D_y}
u
\ln
\frac{a(x)+b(x)u^2}
{c(x)+d(x)u^2}
\,du \notag\\
&\quad+
\int_{D_y}^{2D_y}
(2D_y-u)
\ln
\frac{a(x)+b(x)u^2}
{c(x)+d(x)u^2}
\,du
\Bigg]dx .
\end{align}
Using the definitions of $P_0(y;x)$ and $P_1(y;x)$, the inner integral can be simplified as
\begin{align}
&\int_0^{D_y}
u
\ln
\frac{a(x)+b(x)u^2}
{c(x)+d(x)u^2}
\,du \notag\\
&+
\int_{D_y}^{2D_y}
(2D_y-u)
\ln
\frac{a(x)+b(x)u^2}
{c(x)+d(x)u^2}
\,du \notag\\
&=
P_1(D_y;x)
+
2D_y\big[P_0(2D_y;x)-P_0(D_y;x)\big] \notag\\
&\quad
-\big[P_1(2D_y;x)-P_1(D_y;x)\big] \notag\\
&=
2P_1(D_y;x)-P_1(2D_y;x) \notag \\
&\quad +2D_y\big[P_0(2D_y;x)-P_0(D_y;x)\big] \notag\\
&=
\Phi(x).
\end{align}
Therefore,
\begin{align}
R_{\mathrm{avg}}
=
\frac{1}{D_xD_y^2\ln2}
\int_0^{D_x}\Phi(x)\,dx.
\end{align}

The expressions of $P_0(y;x)$ and $P_1(y;x)$ follow from the standard logarithmic integral identities
\begin{align}
\int \ln(a+bu^2)\,du
&=
u\ln(a+bu^2)-2u \notag \\
& \quad +2\sqrt{\frac{a}{b}}
\arctan\!\left(u\sqrt{\frac{b}{a}}\right),\\
\int u\ln(a+bu^2)\,du
&=
\frac{(a+bu^2)\ln(a+bu^2)-(a+bu^2)}{2b},
\end{align}
which lead to the definitions of $\mathcal{J}_0(\cdot)$ and $\mathcal{J}_1(\cdot)$ in the lemma.

Finally, applying the transformation
\begin{align}
x=\frac{D_x}{2}(t+1), \qquad dx=\frac{D_x}{2}dt,
\end{align}
we obtain
\begin{align}
R_{\mathrm{avg}}
=
\frac{1}{2D_y^2\ln2}
\int_{-1}^{1}
\Phi\!\left(\frac{D_x}{2}(t+1)\right)dt.
\end{align}
Using the Gaussian--Chebyshev quadrature \cite{kumar2022performance} rule with $t_k=\cos\frac{(2k-1)\pi}{2N}$, we obtain \eqref{eq:WDMA_rate_GC}.
\end{proof}

\subsubsection{High-SNR Asymptotic Analysis of WDMA}
To provide further theoretical insight, the high-SNR behavior of the WDMA-enabled PASS is summarized in the following lemma.

\begin{lemma}
\label{lem:WDMA_asym}
The outage probability admits the high-SNR approximation
\begin{align}
P_{\mathrm{out}}^{\infty}
\approx
\frac{\pi}{2N}
\sum_{k=1}^{N}
\sqrt{1-t_k^2}\,
F_Y\!\left(
(\gamma_{\rm th}-1)
G\!\left(\frac{D_x}{2}(t_k+1)\right)
\right),
\label{eq:WDMA_asym_OP}
\end{align}
where $t_k=\cos\frac{(2k-1)\pi}{2N}$ and $N$ is the number of quadrature nodes.
The average achievable rate converges to the finite ceiling
\begin{align}
R_{\mathrm{avg}}^{\infty}
\approx
\frac{\pi}{2ND_y^2\ln2}
\sum_{k=1}^{N}
\sqrt{1-t_k^2}\,
\bar{\Phi}\!\left(\frac{D_x}{2}(t_k+1)\right),
\label{eq:WDMA_asym_rate_single}
\end{align}
where
\begin{align}
\bar{\Phi}(x)
&\triangleq
2\bar{\mathcal{P}}_1(D_y;x)-\bar{\mathcal{P}}_1(2D_y;x) \notag\\
&\quad
+2D_y\big[\bar{\mathcal{P}}_0(2D_y;x)-\bar{\mathcal{P}}_0(D_y;x)\big],
\end{align}
and
\begin{align}
\bar{\mathcal{P}}_n(y;x)
\triangleq
\int_0^y u^n
\ln\!\left(2+\frac{u^2}{G(x)}\right)du,\quad n\in\{0,1\}.
\end{align}
\end{lemma}

\begin{proof}
Letting $P\to\infty$ in \eqref{eq:WDMA_SINR}, the noise term vanishes while the inter-waveguide interference remains. Therefore,
\begin{align}
\gamma_1^{\infty}
&=
\frac{|g_1|^2}{|g_{1,2}|^2} \notag\\
&=
\frac{\eta/G(x)}{\eta/(G(x)+Y)}
=
1+\frac{Y}{G(x)}.
\end{align}
The corresponding outage event is
\begin{align}
\gamma_1^{\infty}\le \gamma_{\rm th}
\quad \Longleftrightarrow \quad
Y\le (\gamma_{\rm th}-1)G(x).
\end{align}
Averaging over $x_{u,1}$ gives
\begin{align}
P_{\mathrm{out}}^{\infty}
=
\frac{1}{D_x}
\int_0^{D_x}
F_Y\!\big((\gamma_{\rm th}-1)G(x)\big)\,dx.
\end{align}
By applying the transformation $x=\frac{D_x}{2}(t+1)$, we have
\begin{align}
P_{\mathrm{out}}^{\infty}
=
\frac{1}{2}
\int_{-1}^{1}
F_Y\!\left(
(\gamma_{\rm th}-1)
G\!\left(\frac{D_x}{2}(t+1)\right)
\right)dt.
\end{align}
Using Gaussian--Chebyshev quadrature \cite{kumar2022performance} yields \eqref{eq:WDMA_asym_OP}.

For the average achievable rate, taking $B\to0$ in the instantaneous rate expression of UE~1 gives
\begin{align}
R_1^{\infty}
&=
\log_2\!\left(
1+\frac{U^2+G(x)}{G(x)}
\right) \notag\\
&=
\log_2\!\left(
2+\frac{U^2}{G(x)}
\right).
\end{align}
Using the triangular distribution of $U$, the average rate ceiling can be written as
\begin{align}
R_{\mathrm{avg}}^{\infty}
&=
\frac{1}{D_xD_y^2\ln2}
\int_0^{D_x}
\Bigg[
\int_0^{D_y}
u\ln\!\left(2+\frac{u^2}{G(x)}\right)du \notag\\
&\quad+
\int_{D_y}^{2D_y}
(2D_y-u)
\ln\!\left(2+\frac{u^2}{G(x)}\right)du
\Bigg]dx .
\end{align}
Using the definitions of $\bar{\mathcal{P}}_0(y;x)$ and 
$\bar{\mathcal{P}}_1(y;x)$, the inner integral becomes
\begin{align}
&\int_0^{D_y}
u\ln\!\left(2+\frac{u^2}{G(x)}\right)du
+
\int_{D_y}^{2D_y}
(2D_y-u)
\ln\!\left(2+\frac{u^2}{G(x)}\right)du \notag\\
&=
2\bar{\mathcal{P}}_1(D_y;x)
-\bar{\mathcal{P}}_1(2D_y;x)
+
2D_y\big[
\bar{\mathcal{P}}_0(2D_y;x)
-\bar{\mathcal{P}}_0(D_y;x)
\big] \notag\\
&=
\bar{\Phi}(x).
\end{align}
Thus,
\begin{align}
R_{\mathrm{avg}}^{\infty}
=
\frac{1}{D_xD_y^2\ln2}
\int_0^{D_x}
\bar{\Phi}(x)\,dx.
\end{align}
Applying $x=\frac{D_x}{2}(t+1)$ gives
\begin{align}
R_{\mathrm{avg}}^{\infty}
=
\frac{1}{2D_y^2\ln2}
\int_{-1}^{1}
\bar{\Phi}\!\left(\frac{D_x}{2}(t+1)\right)dt.
\end{align}
Using Gaussian--Chebyshev quadrature \cite{kumar2022performance} yields \eqref{eq:WDMA_asym_rate_single}.
\end{proof}

Lemma~\ref{lem:WDMA_asym} shows that WDMA becomes interference-limited at high SNR, leading to a generally non-zero outage floor when $\gamma_{\rm th}>1$ and a finite rate ceiling.


\subsection{NOMA}
Next, we characterize the downlink performance of the considered PASS under NOMA. We first derive the outage probabilities of UE 1 and UE 2, and then present the corresponding average achievable rates.

\subsubsection{Outage Probability}
The outage events of the two NOMA users can be first rewritten into tractable geometric forms. 
For UE~1, the outage probability under a target threshold $\gamma_{\rm th}$ is
\begin{align}
P_{{\rm out},1}
&=\Pr\!\left\{\frac{\eta\alpha_1P}{\sigma_1^2\big((x_{u,1}-D_x/2)^2+h^2\big)}\le \gamma_{\rm th}\right\} \notag \\
&=\Pr\!\left\{G_1 \ge C_1\right\},
\end{align}
where $G_1=(x_{u,1}-D_x/2)^2$ and 
$C_1\triangleq \frac{\eta\alpha_1P}{\gamma_{\rm th}\sigma_1^2}-h^2$.
For UE~2, the outage probability can be expressed as
\begin{align}
P_{{\rm out},2}
=
\Pr\!\left\{(x_{u,2}-D_x/2)^2+(y_{u,2}-y_p)^2 \ge C_2\right\},
\end{align}
where 
$C_2 \triangleq \frac{\eta\alpha_2P}{\gamma_{\rm th}\sigma_2^2}
-\frac{\eta\alpha_1P}{\sigma_2^2}-h^2$.

\begin{lemma}
The outage probability of UE~1 is given by
\begin{align}
P_{{\rm out},1}=
\begin{cases}
1, & C_1\le 0,\\[4pt]
1-\dfrac{4\sqrt{C_1}}{D_x}+\dfrac{4C_1}{D_x^2}, 
& 0<C_1<\big(\tfrac{D_x}{2}\big)^2,\\[10pt]
0, & C_1\ge \big(\tfrac{D_x}{2}\big)^2.
\end{cases}
\label{eq:UE 1-outage}
\end{align}
The outage probability of UE~2 is given by
\begin{align}
P_{{\rm out},2}
= \frac{4}{D_x^{2}}
\Big(\big[\Phi_1(m)\big]_{M_1}^{M_2}
+ \big[\Phi_2(m)\big]_{M_2}^{M_3}
+ \big[\Phi_3(m)\big]_{M_3}^{M_4}\Big),
\label{eq:UE 2-outage}
\end{align}
where
\begin{align}
M_4&=\Big(\tfrac{D_x}{2}\Big)^2,\quad
M_1=\big[C_2-4D_y^2\big]_{0}^{M_4}, \notag\\
M_2&=\big[C_2-D_y^2\big]_{0}^{M_4},\quad
M_3=\big[C_2\big]_{0}^{M_4},
\end{align}
with $[z]_{0}^{M_4}\triangleq \min\{\max(z,0),M_4\}$, and
\begin{align}
\Phi_1(m)
&=2m+\frac{4}{3D_y}(C_2-m)^{3/2}
+\frac{1}{2D_y^{2}}\!\left(C_2m-\frac{m^{2}}{2}\right),\\
\Phi_2(m)
&=m-\frac{C_2m}{2D_y^{2}}+\frac{m^{2}}{4D_y^{2}},\\
\Phi_3(m)&=m.
\end{align}
In addition, if $C_2\le0$, then $P_{{\rm out},2}=1$; if 
$C_2\ge (D_x/2)^2+(2D_y)^2$, then $P_{{\rm out},2}=0$.
\end{lemma}

\begin{proof}
For UE~1, the outage event has been rewritten as
\begin{align}
P_{{\rm out},1}=\Pr\{G_1\ge C_1\}.
\end{align}
Since UE~1 is the user closer to the waveguide center, for
$0<g<\big(D_x/2\big)^2$, the event $G_1\le g$ means that at least one of the two candidate horizontal coordinates lies within the interval
$\big[D_x/2-\sqrt{g},D_x/2+\sqrt{g}\big]$. Hence,
\begin{align}
F_{G_1}(g)
&=1-\left(1-\frac{2\sqrt{g}}{D_x}\right)^2 \notag\\
&=\frac{4\sqrt{g}}{D_x}-\frac{4g}{D_x^2}.
\end{align}
Therefore,
\begin{align}
P_{{\rm out},1}
=
1-F_{G_1}(C_1),
\end{align}
which directly yields the piecewise expression in \eqref{eq:UE 1-outage}.

For UE~2, conditioning on $x_{u,2}=x$ gives
\begin{align}
P_{{\rm out},2|x}
=
\Pr\!\left\{
(y_{u,2}-y_p)^2
\ge
C_2-(x-D_x/2)^2
\right\}.
\end{align}
Averaging over the ordered horizontal coordinate of UE~2 gives
\begin{align}
P_{{\rm out},2}
=
\int_0^{D_x}
P_{{\rm out},2|x}
\frac{4|D_x/2-x|}{D_x^2}
\,dx.
\end{align}
Since $A=|y_{u,2}-y_p|$ has the same distribution as $U$ in the WDMA case, its CDF follows the same piecewise form and is omitted here for brevity. Let $r(x)\triangleq \sqrt{\max\{0,\,C_2-(x-D_x/2)^2\}}$, then
\begin{align}
P_{{\rm out},2|x}=
\begin{cases}
1, & C_2\le (x-D_x/2)^2,\\[3pt]
1-F_A\!\big(r(x)\big), & 0<r(x)\le 2D_y,\\[3pt]
0, & r(x)>2D_y.
\end{cases}
\end{align}
Applying the change of variable $m=(x-D_x/2)^2$, the above integral becomes
\begin{align}
P_{{\rm out},2}
=
\frac{4}{D_x^2}
\int_0^{M_4} g(m)\,dm,
\end{align}
where $g(m)$ denotes the corresponding conditional outage probability after the change of variable.

Since $A=|y_{u,2}-y_p|$ has the same distribution as $U$ in the WDMA case, the integral over $m$ can be divided according to the same triangular distribution as
\begin{align}
P_{{\rm out},2}
&=
\frac{4}{D_x^2}
\Bigg[
\int_{M_1}^{M_2}
\left(
2-\frac{2\sqrt{C_2-m}}{D_y}
+\frac{C_2-m}{2D_y^2}
\right)dm \notag\\
&\quad+
\int_{M_2}^{M_3}
\left(
1-\frac{C_2-m}{2D_y^2}
\right)dm
+
\int_{M_3}^{M_4}1\,dm
\Bigg].
\end{align}
Taking the antiderivatives of the three integrands gives
\begin{align}
P_{{\rm out},2}
=
\frac{4}{D_x^2}
\Big(
[\Phi_1(m)]_{M_1}^{M_2}
+
[\Phi_2(m)]_{M_2}^{M_3}
+
[\Phi_3(m)]_{M_3}^{M_4}
\Big),
\end{align}
which proves \eqref{eq:UE 2-outage}.

Finally, if $C_2\le0$, the event 
$(x_{u,2}-D_x/2)^2+(y_{u,2}-y_p)^2\ge C_2$ always holds, and thus
$P_{{\rm out},2}=1$. If 
$C_2\ge (D_x/2)^2+(2D_y)^2$, the event only occurs on a boundary set with zero probability, and thus $P_{{\rm out},2}=0$.
\end{proof}

\subsubsection{Average Achievable Rate}
Let $K \triangleq \eta\alpha_1P/\sigma_1^2$ and $c=D_x/2$. 
Since the PDF of $x_{u,1}$ is
$f_{x_{u,1}}(x)=2/D_x-4|x-c|/D_x^2$, $0\le x\le D_x$, the average achievable rate of UE~1 is
\begin{align}
R_{{\rm avg},1}
&=\int_0^{D_x}\!\log_2\!\Big(1+\frac{K}{G(x)}\Big)f_{x_{u,1}}(x)dx \notag \\
&=\frac{1}{\ln 2}\Bigg[
\frac{4}{D_x}\big(J_0(c; h^2{+}K, 1)-J_0(c; h^2, 1)\big)\notag\\
&\quad-\frac{8}{D_x^2}\big(J_1(c; h^2{+}K, 1)-J_1(c; h^2, 1)\big)\Bigg].
\end{align}

For UE~2, let $K_1\triangleq \eta\alpha_1P$, 
$K_2\triangleq \eta\alpha_2P$, $N_2\triangleq\sigma_2^2$, 
$M\triangleq(x_{u,2}-c)^2$, 
$\beta_1(m)=K_1+N_2(h^2+m)$, and 
$\beta_0(m)=\beta_1(m)+K_2$. 
The average achievable rate of UE~2 is
\begin{align}
R_{{\rm avg},2}
=
\frac{4}{D_x^2\ln2}
\int_0^{c^2}
\Delta H(m)\,dm,
\end{align}
where $\Delta H(m)\triangleq H(\beta_0(m))-H(\beta_1(m))$, and
\begin{align}
H(\beta)
&= \frac{\mathcal{J}_1(D_y;\beta,N_2)-\mathcal{J}_1(0;\beta,N_2)}{D_y^2}\notag \\
&\quad+\frac{2\big[\mathcal{J}_0(2D_y;\beta,N_2)-\mathcal{J}_0(D_y;\beta,N_2)\big]}{D_y} \notag\\
&\quad
-\frac{\mathcal{J}_1(2D_y;\beta,N_2)-\mathcal{J}_1(D_y;\beta,N_2)}{D_y^2}.
\end{align}
By applying Gaussian--Chebyshev quadrature \cite{kumar2022performance}, $R_{{\rm avg},2}$ can be approximated as
\begin{align}
R_{{\rm avg},2}
\approx
\frac{2c^2\pi}{D_x^2 N\ln2}
\sum_{k=1}^{N}\sqrt{1-x_k^2}\,
\Delta H(m_k),
\label{eq:UE 2-rate}
\end{align}
where $x_k=\cos\frac{(2k-1)\pi}{2N}$, 
$m_k=\frac{c^2}{2}(x_k+1)$, and $N$ is the number of quadrature nodes.

\subsubsection{High-SNR Asymptotic Analysis of NOMA}
Letting $P\to\infty$ yields the following high-SNR results. 
The outage probabilities of UE~1 and UE~2 become zero once
\begin{align}
P &\ge 
\frac{\gamma_{\rm th}\sigma_1^2\big[(D_x/2)^2+h^2\big]}
{\eta\alpha_1}, \label{eq:UE 1-threshold}\\
P &\ge 
\frac{\gamma_{\rm th}\sigma_2^2\big[(D_x/2)^2+(2D_y)^2+h^2\big]}
{\eta(\alpha_2-\gamma_{\rm th}\alpha_1)},
\label{eq:UE 2-threshold}
\end{align}
respectively, where the second threshold exists only when 
$\alpha_2>\gamma_{\rm th}\alpha_1$; otherwise, UE~2 is always in outage. 
For the average rates, UE~1 achieves full multiplexing gain,
i.e., $R_{{\rm avg},1}=\log_2P+\mathcal{O}(1)$, while UE~2 approaches the power-allocation-limited ceiling
\begin{align}
R_{{\rm avg},2}\to \log_2\!\left(1+\frac{\alpha_2}{\alpha_1}\right).
\end{align}
These results show that NOMA is noise-limited for UE~1 after SIC, but interference-limited for UE~2 at high SNR.


\section{Numerical results}
We validate the analysis by Monte Carlo simulations. Unless otherwise specified, the parameters are as follows. The carrier frequency is $f_c = 28~\mathrm{GHz}$, yielding $\lambda_c = c_l/f_c$ and the free-space path loss factor 
$\eta = \lambda_c^2/(16\pi^2)$. The total noise power is $\sigma^2 = -90~\mathrm{dBm}$, and the waveguide height is fixed at $h = 3~\mathrm{m}$. The outage threshold is set to $\gamma_{\mathrm{th}} = 5$, and the user region is a rectangle of size $D_x \times D_y = 10 \times 20~\mathrm{m}^2$. For NOMA transmission, the power allocation factors are $\alpha_1 = 0.05$ and $\alpha_2 = 0.95$. All outage probabilities are averaged over $10^5$ independent channel realizations, while the analytical integrals are evaluated using Gaussian--Chebyshev quadrature with $N=64$ nodes. The transmit SNR is defined as $\gamma_t = P_t/\sigma^2$. 

\begin{figure}[!t]
\centering
\subfloat[Outage probability.]{%
  \includegraphics[width=0.50\linewidth]{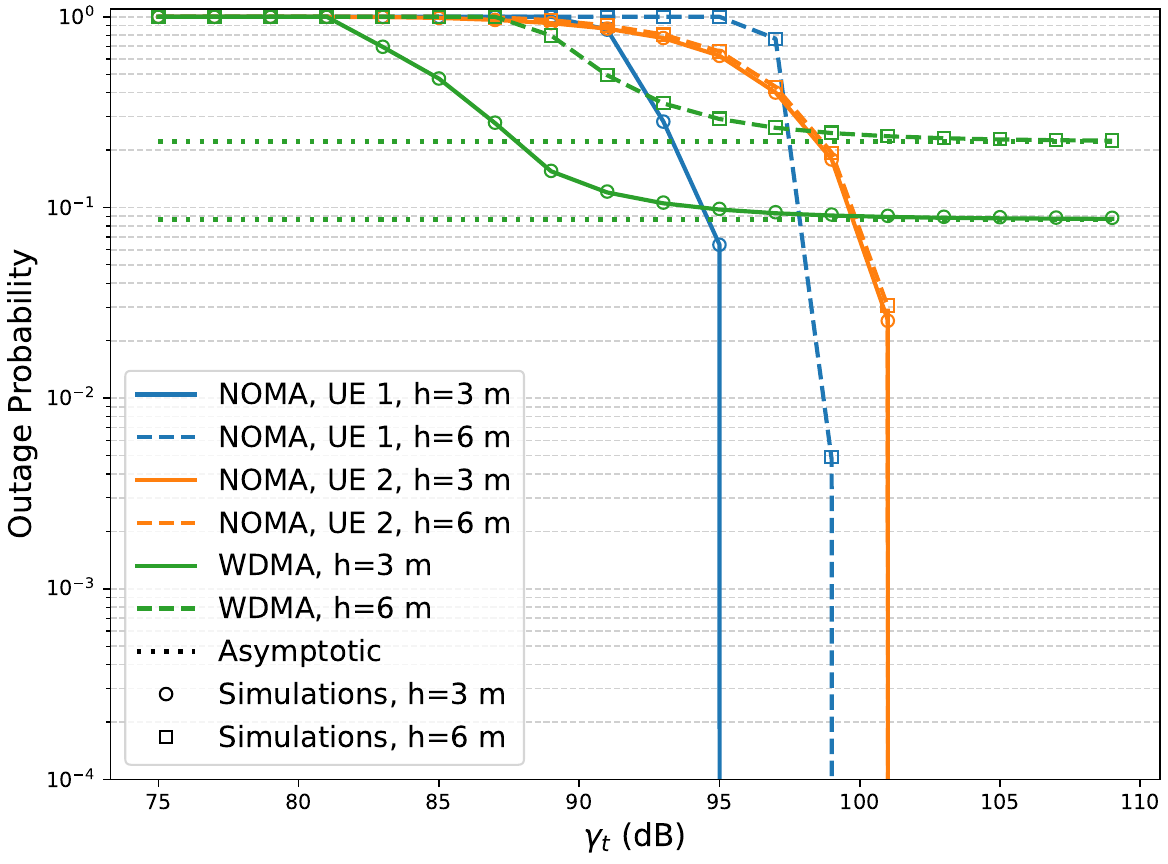}%
  \label{fig:op_snr_h}}
\hfill
\subfloat[Average achievable rate.]{%
  \includegraphics[width=0.50\linewidth]{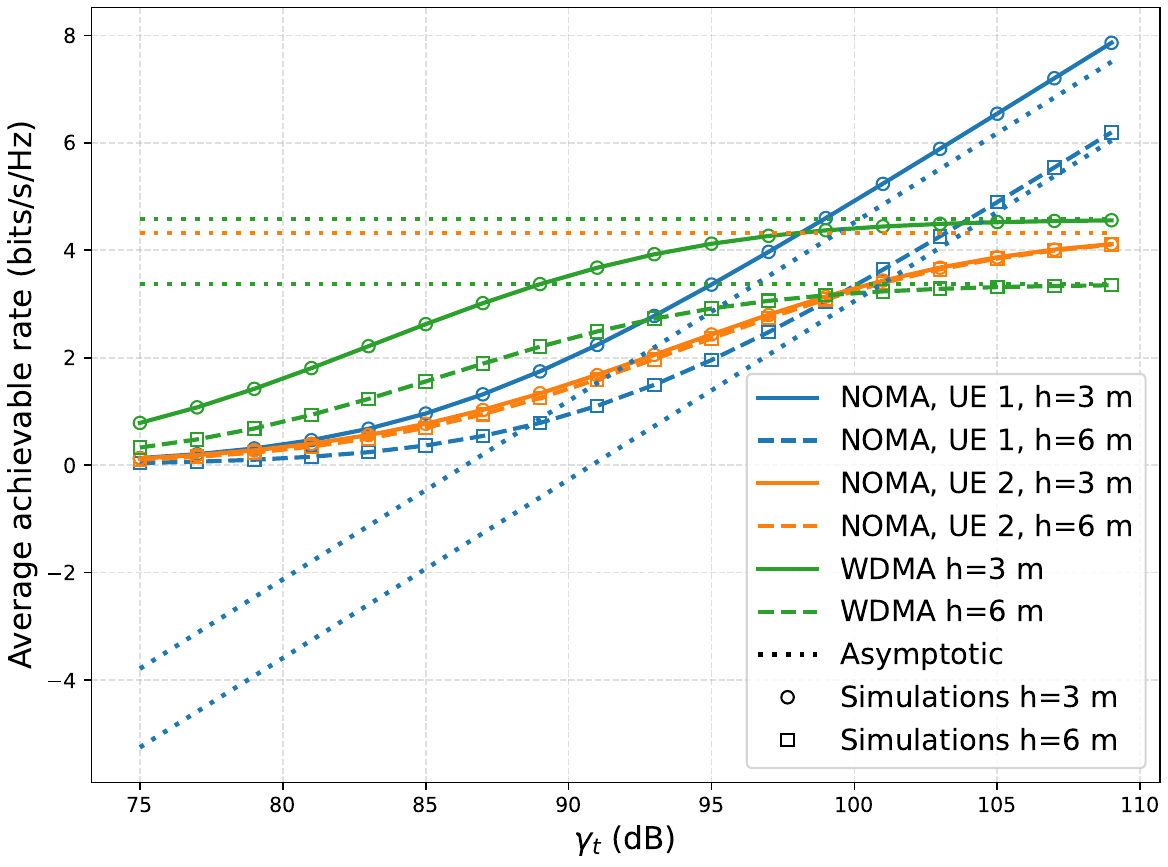}%
  \label{fig:rate_snr_h}}
\caption{Performance versus transmit SNR for different PA heights.}
\label{fig:perf_snr_h}
\end{figure}

Fig.~\ref{fig:perf_snr_h} shows the outage probability and average achievable rate versus transmit SNR for different PA heights. Increasing $h$ from $3$~m to $6$~m degrades both reliability and rate performance due to stronger large-scale path loss. NOMA UE~1 achieves the fastest rate growth at high SNR owing to SIC, whereas WDMA exhibits an outage floor and rate saturation caused by persistent inter-waveguide interference. 
In particular, WDMA converges to the derived outage floors and rate ceilings, while NOMA UE~2 approaches the power allocation limited rate ceiling $\log_2(1+\alpha_2/\alpha_1)$.

Next, we consider two spatial configurations, denoted by $\Omega_1$ and $\Omega_2$. In $\Omega_1$, both UEs are uniformly distributed with $D_y = 10\,\mathrm{m}$, representing a compact deployment. In $\Omega_2$, the UE 1 position follows $y_{\mathrm{u},1} \sim \mathcal{U}(10,\,20)\,\mathrm{m}$ and the UE 2 position follows $y_{\mathrm{u},2} \sim \mathcal{U}(-20,\,-10)\,\mathrm{m}$, corresponding to a more dispersed UE placement.

\begin{figure}[!t]
\centering
\subfloat[Outage probability.]{%
  \includegraphics[width=0.50\linewidth]{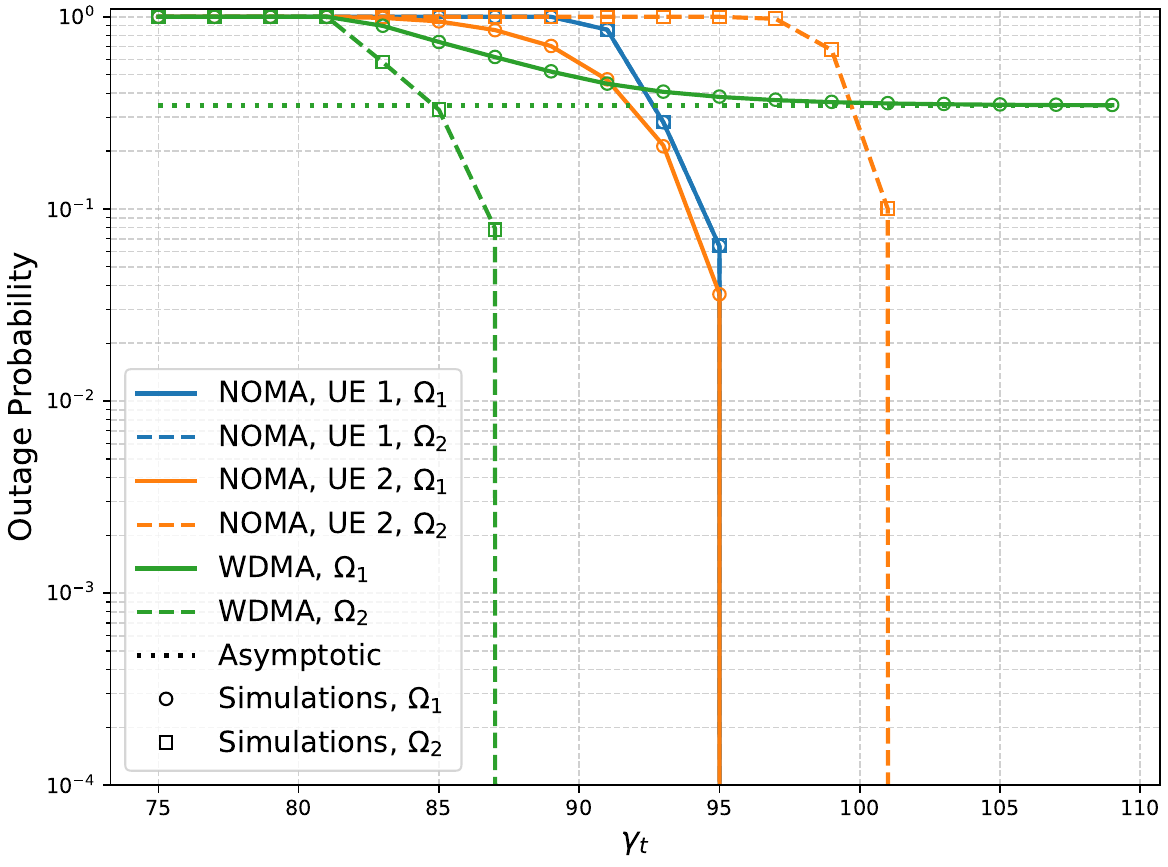}%
  \label{fig:op_snr_spatial}}
\hfill
\subfloat[Average achievable rate.]{%
  \includegraphics[width=0.50\linewidth]{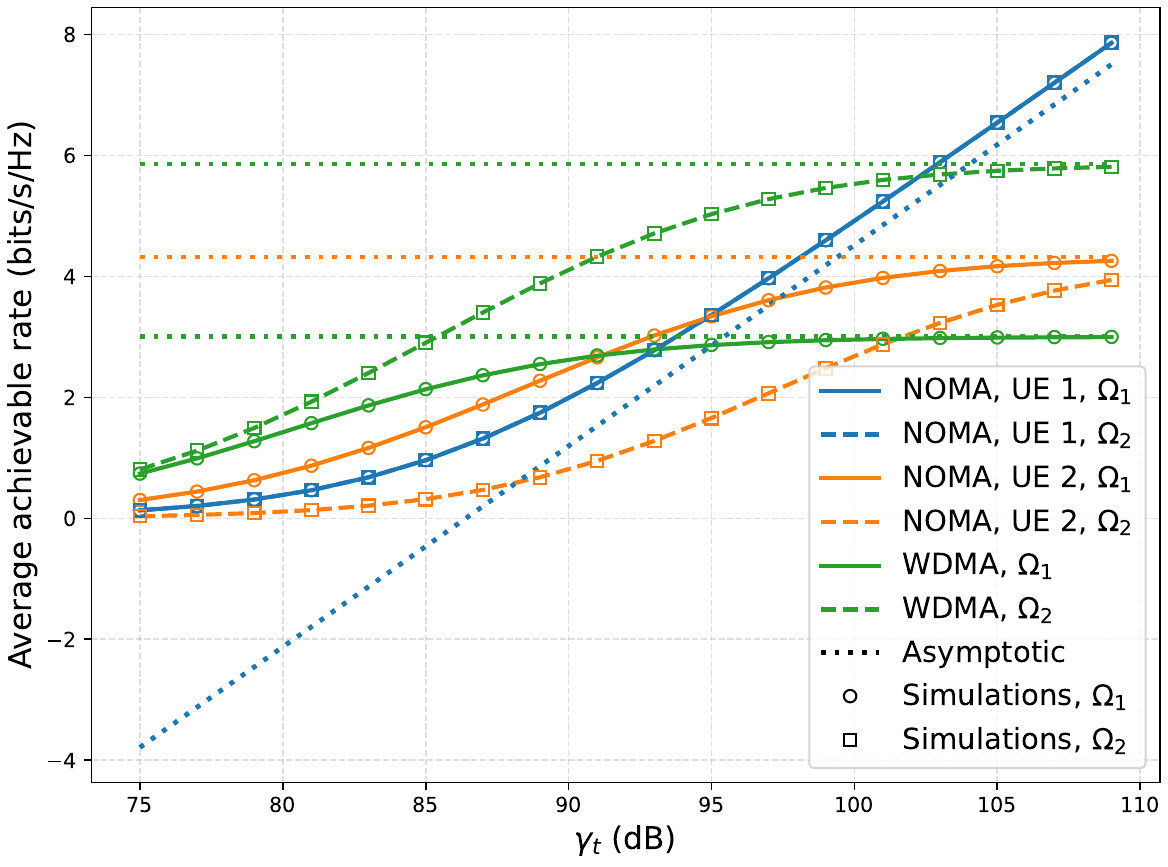}%
  \label{fig:rate_snr_spatial}}
\caption{Performance comparison of NOMA and WDMA under two spatial configurations $\Omega_1$ and $\Omega_2$.}
\label{fig:perf_spatial}
\end{figure}

Fig.~\ref{fig:perf_spatial} compares NOMA and WDMA under the two spatial configurations. In the compact deployment $\Omega_1$, WDMA suffers from an evident outage floor and rate saturation, while NOMA achieves higher rates at high transmit SNR. When the UE separation increases in $\Omega_2$, the inter-waveguide interference is significantly reduced, and WDMA achieves improved reliability and rate performance. In contrast, NOMA UE~2 degrades in $\Omega_2$ because the increased distance from the PA causes stronger path loss. 
The asymptotic curves further confirm that increasing the UE region separation reduces the WDMA outage floor and raises its rate ceiling, which is consistent with the derived dependence on the inter waveguide separation term $Y$.

\begin{figure}[!t]
\centering
\subfloat[Outage probability.]{%
  \includegraphics[width=0.50\linewidth]{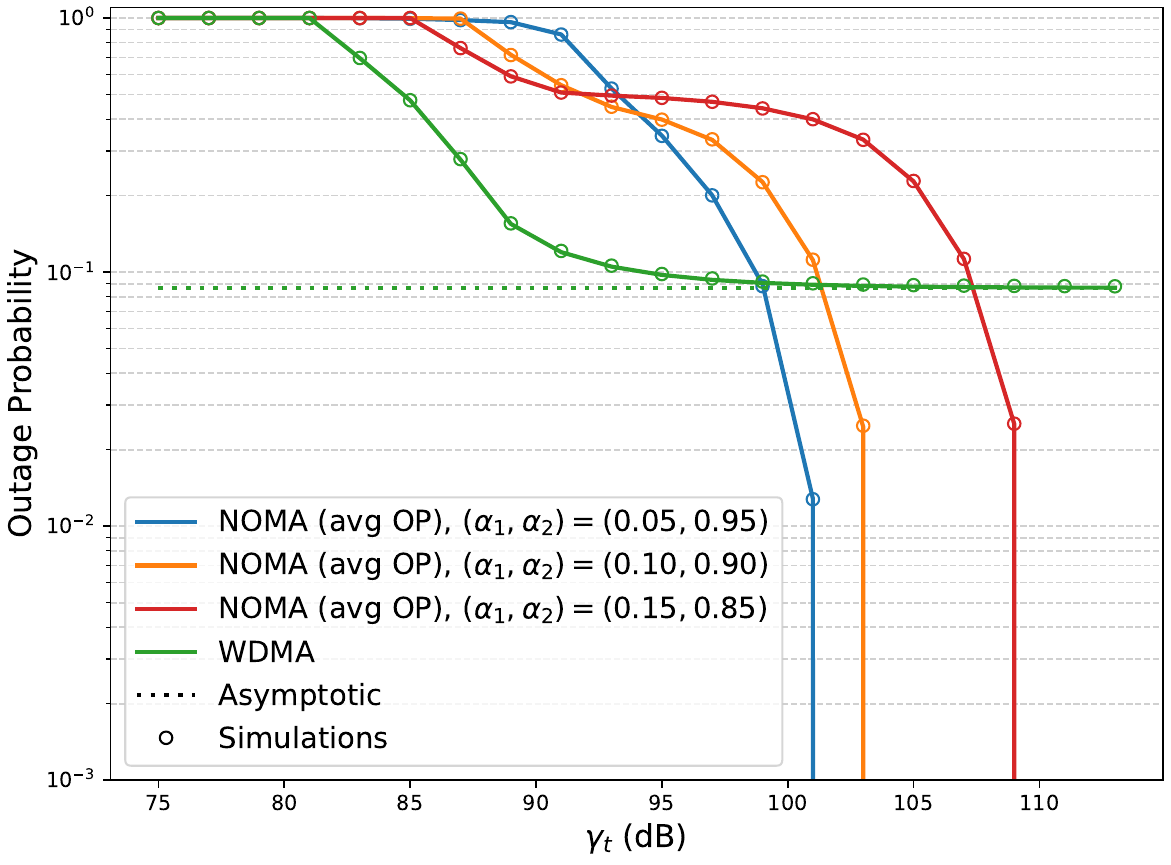}%
  \label{fig:op_snr_alpha}}
\hfill
\subfloat[Average achievable rate.]{%
  \includegraphics[width=0.50\linewidth]{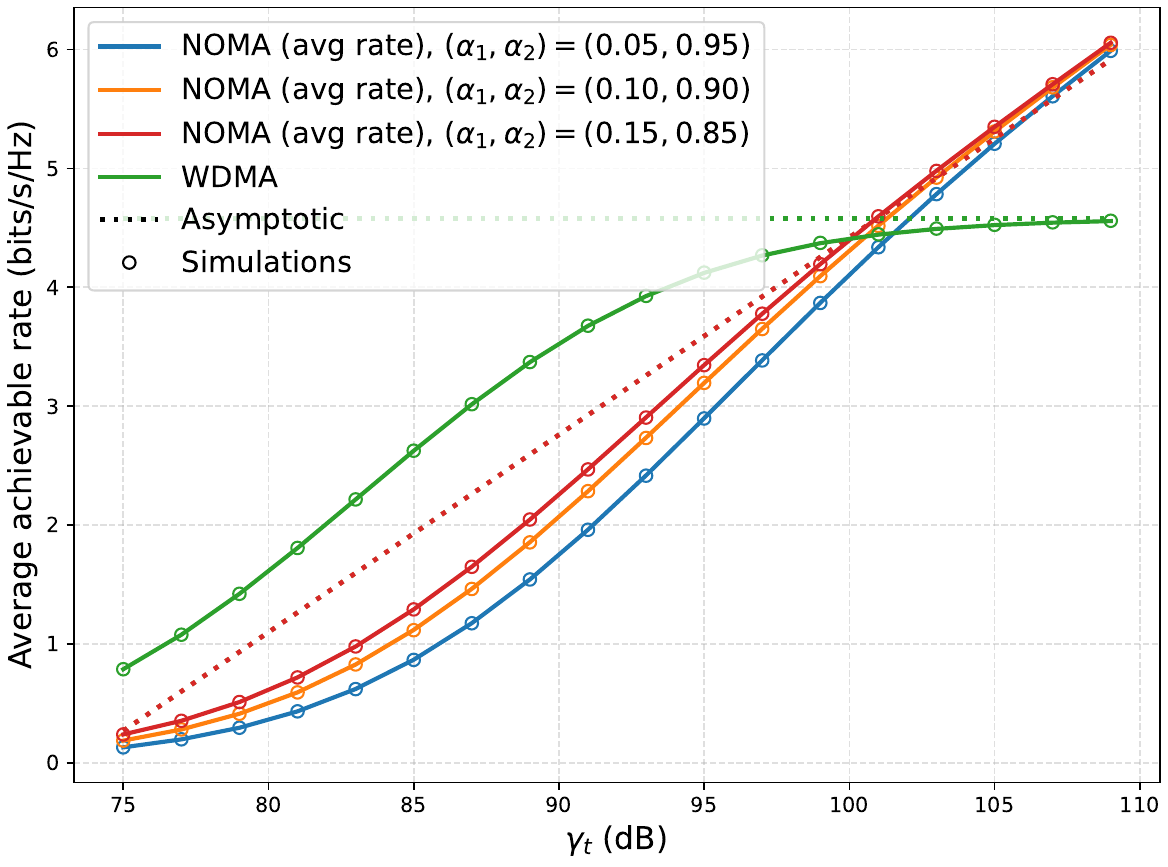}%
  \label{fig:rate_snr_alpha}}
\caption{Performance comparison of WDMA and NOMA under different NOMA power allocation settings $(\alpha_1,\alpha_2)$.}
\label{fig:perf_alpha}
\end{figure}

Fig.~\ref{fig:perf_alpha} shows the impact of NOMA power allocation. WDMA performs better in the low-SNR regime, while NOMA becomes superior at sufficiently high SNR. Increasing $\alpha_1$ improves the near-user rate after SIC, thus shifting the rate intersection point toward lower SNR; however, it also reduces the power allocated to the far user and pushes the outage intersection point toward higher SNR. 
The asymptotic results further explain this behavior. WDMA converges to floors or ceilings determined by geometry, while the average rate of the two NOMA users has the same trend at high SNR under different power allocations because the terms related to $\alpha_1$ cancel out. Hence, power allocation mainly affects the transition region at finite SNR and the outage performance.

Overall, these results indicate that the preferred multiple-access scheme depends on the SNR regime, UE spatial distribution, and NOMA power allocation. WDMA is attractive at low-to-moderate SNR or with sufficiently separated UE regions, whereas NOMA provides higher spectral efficiency at high SNR due to SIC.

\section{Conclusion}

This paper developed an analytical framework for downlink PASS with WDMA and NOMA, deriving tractable expressions for the outage probability and average achievable rate. The results show that WDMA is preferable at low transmit SNR and with larger UE separation, while NOMA achieves higher rates at high SNR via successive interference cancellation. Moreover, increasing the PA–UE vertical distance degrades the performance of both schemes. These results provide guidelines for selecting an appropriate multiple-access strategy in PASS under different transmit powers and antenna deployments.

\bibliographystyle{IEEEtran}
\bibliography{MyRef}

\end{document}